\documentclass[11pt,aps,pra,onecolumn,superscriptaddress,floatfix,
nofootinbib,showpacs,longbibliography]{revtex4-2}
\usepackage[utf8]{inputenc}  
\usepackage[T1]{fontenc}     
\usepackage[british]{babel}  
\usepackage[sc,osf]{mathpazo}
\usepackage{libertineRoman}  
\usepackage[colorlinks=true, citecolor=blue, urlcolor=blue]{hyperref}  
\usepackage{graphicx}
\usepackage{diagbox}
 \usepackage{enumitem}
 \usepackage[babel]{microtype}  
\usepackage{amsmath,amssymb,amsthm,bm,amsfonts,mathrsfs,bbm} 

\usepackage{xspace}  
\usepackage{pgf,tikz}
\usepackage{xcolor}
\usepackage{multirow}
\usepackage{array}
\usepackage{bigstrut}
\usepackage{braket}
\usepackage{color}
\usepackage{natbib}
\usepackage{multirow}
\usepackage{mathtools}
\usepackage{float}
\usepackage{xcolor,colortbl}
\usepackage{physics}
\usepackage{amsmath}
\usepackage{color}
\usepackage[justification=justified, format=plain]{subcaption}
\usepackage[justification=raggedright]{caption}

\newcommand{\be}{\begin{equation}}
\newcommand{\ee}{\end{equation}}
\newcommand{\ba}{\begin{eqnarray}}
\newcommand{\ea}{\end{eqnarray}}

\newtheorem{theorem}{Theorem}

\newtheorem{definition}{Definition}

\newtheorem{example}{Example}
\newtheorem{remark}{Remark}
\newtheorem{lemma}{Lemma}

\def\>{\rangle}
\def\<{\langle}

\begin{document}

\title{Certifying the dimensionality of any quantum channel with minimal assumptions}	

\author{Saheli Mukherjee}
\email{mukherjeesaheli95@gmail.com}
\affiliation{S. N. Bose National Centre for Basic Sciences, Block JD, Sector III, Salt Lake, Kolkata 700 106, India}

\author{Bivas Mallick}
\email{bivasqic@gmail.com}
\affiliation{S. N. Bose National Centre for Basic Sciences, Block JD, Sector III, Salt Lake, Kolkata 700 106, India}

\author{Pratik Ghosal}
\email{ghoshal.pratik00@gmail.com}
\thanks{Present address: Harish-Chandra Research Institute,
Chhatnag Road, Jhunsi, Allahabad 211019, India; Homi Bhabha
National Institute, Training School Complex, Anushakti Nagar,
Mumbai 400094, India.}
\affiliation{S. N. Bose National Centre for Basic Sciences, Block JD, Sector III, Salt Lake, Kolkata 700 106, India}

\begin{abstract}
High-dimensional entanglement offers significant advantages over its low-dimensional counterpart in various information-processing tasks. However, to harness these advantages, it is crucial that the quantum channels used to store or transmit the subsystems of an entangled system not only preserve entanglement but also maintain its dimensionality above a certain threshold. The maximum entanglement dimension that a channel can preserve is referred to as its \textit{effective dimensionality}, since the channel cannot be used to transmit information of dimension greater than that in a single use. In this work, we present a method to certify whether a quantum channel can preserve entanglement dimension above a given threshold. Unlike existing approaches, our method is \textit{faithful}, i.e., it can be applied to any channel, and avoids common assumptions such as preparation of entangled states, auxiliary side channels, or perfect measurement devices. Moreover, the method can be extended to faithfully certify other classes of non-resource-breaking channels, such as non-nonpositive-partial-transpose- breaking channels (non-NPT-breaking channels). Finally, we discuss possible experimental realizations of our certification scheme through explicit examples.

\end{abstract}
\maketitle

\section{Introduction}

The program of the `Second Quantum Revolution' aims to harness the fundamental properties of quantum systems to develop technologies with transformative applications in computing, secure communication, and precision sensing \cite{Dowling2003, Deutsch2020}. However, unavoidable interactions of quantum systems with their environment makes them highly susceptible to noise and decoherence. Therefore, to realize quantum technologies in practical settings, it is critical that the devices used to store (quantum memories) and transmit (communication links) quantum systems preserve their essential quantum features that serve as resources for various tasks. Since such devices are mathematically modeled as quantum channels \cite{nielsen2010quantum}, characterizing these channels based on their specific resource-breaking properties \cite{Horodecki2003, pal2015non, heinosaari2015incompatibility, luo2022coherence, ku2022quantifying, muhuri2023information, srinidhi2024quantum,kumar2025fidelity,patra2024qubit,devendra2023mapping,mallick2025higher} and certifying those that are non-resource-breaking is key to determining their suitability for such tasks.

In recent years, the experimental demonstration of control over high-dimensional entangled systems \cite{Dada2011, Malik2016} has sparked a major research focus on efficiently utilizing these systems for secure quantum communication. Unlike the local Hilbert space dimensions of an entangled system, the \textit{dimension of entanglement} quantifies the number of degrees of freedom genuinely involved in the entanglement. In bipartite systems, it is commonly characterized by the Schmidt number of the state \cite{terhal2000schmidt}. High-dimensional entanglement offers several advantages over its low-dimensional counterpart, including increased robustness against noise \cite{PhysRevA.71.044305,groblacher2006experimental,lanyon2009simplifying,PhysRevA.88.032309,Mirhosseini_2015}, enhanced key rates in quantum key distribution \cite{sit2017high}, improved success probabilities in channel-discrimination tasks \cite{bae2019more}, and greater channel capacities in quantum communication \cite{cozzolino2019high}.

To ensure these advantages, it is essential that the quantum channels used in such protocols are not only non-entanglement-breaking but also capable of preserving entanglement above a certain dimensional threshold. A perfect (noiseless) $d$-dimensional quantum channel can faithfully transmit any state of a $d$-dimensional quantum system and perfectly preserve any entanglement the system shares with external parties. In contrast, a noisy channel typically degrades entanglement, with the most extreme case being entanglement-breaking channels \cite{Horodecki2003}, which have zero quantum capacity and effectively behave as classical channels. When the noise is less severe, a channel may not destroy entanglement entirely but can still reduce its dimensionality. Channels exhibiting this behavior are known as \textit{partially entanglement-breaking} \cite{chruscinski2006partially,johnston2008partially,shirokov2013schmidt} or \textit{Schmidt-number-breaking} channels \cite{mallick2024characterization}. A channel is said to be $k$-\textit{Schmidt-number-breaking} ($k$-SNB) if, when acting on one half of any bipartite state, it produces an output with a Schmidt number at most $k$. As a consequence, its coherent information (or single-shot quantum capacity) is upper bounded by $\log k$, which implies that, despite being $d$-dimensional, such a channel cannot be used to reliably transmit more than $k$-dimensional quantum information in a single use. We therefore refer to the maximum entanglement dimension a channel can preserve as its \textit{effective dimensionality}.

In this paper, we certify whether a quantum channel preserves entanglement of dimension greater than~$k$, \textit{i.e.}, whether it is non-$k$-SNB. It turns out that, for a $d$-dimensional channel, it suffices to verify that its action on one half of a maximally entangled state in $\mathbb{C}^d \otimes \mathbb{C}^d$ yields a bipartite state with Schmidt number strictly greater than $k$. The convexity of the set of states with Schmidt number at most $k$ allows one to test this condition using a suitably chosen Schmidt number witness operator \cite{sanpera2001schmidt,shi2024families}, which yields a strictly negative expectation value if the output state's Schmidt number exceeds $k$, and a non-negative value otherwise. The expectation value is obtained by decomposing the witness into a set of local observables, which are then measured on the individual subsystems of the output bipartite state. However, this approach suffers from a practical drawback: it requires complete trust in both the state preparation and measurement devices, and any imperfection in either can lead to false positives, which is particularly undesirable in practical scenarios.

An alternative approach is to certify the Schmidt number of the output state in a fully device-independent manner using Bell nonlocal games \cite{PhysicsPhysiqueFizika.1.195, RevModPhys.86.419}. In this framework, local measurements chosen according to classical random inputs, are performed on the individual subsystems, producing classical outputs. If the resulting input-output correlations cannot be reproduced by any state with Schmidt number at most $k$, for any choice of local measurements, this certifies that the Schmidt number of the underlying state exceeds $k$. Since it relies solely on observed correlations, this method is inherently immune to errors or imperfections in state preparation and measurement devices.

Nevertheless, this approach is also subject to several important constraints: $(i)$ It cannot certify all channels that are not $k$-SNB. This is because, fully device-independent certification of the Schmidt number is not possible for all entangled states \cite{nlz1-h6qr, hirsch2020schmidt}, even for some that exhibit Bell nonlocality. In \cite{nlz1-h6qr}, we, along with our collaborators, demonstrated a family of states that are Bell nonlocal, yet whose Schmidt number cannot be certified via any Bell nonlocal game when the parties are restricted to local projective or two-outcome generalized measurements (POVMs). $(ii)$ The method requires the reliable preparation of entangled states as input. $(iii)$ It assumes access to an additional (ideally noiseless) side channel to preserve the subsystem on which the test channel does not act. $(iv)$ Finally, Bell tests are not robust against particle losses \cite{RevModPhys.86.419, PhysRevA.46.3646}.

Ideally, one would like to have a certification method that is free from all these constraints. However, that does not seem to be possible in a fully device-independent setting. Lifting constraints~$(ii)$ and~$(iii)$, in particular, necessitates a certification method in a time-like scenario, where a fully device-independent approach is fundamentally impossible. On the other hand, in \cite{nlz1-h6qr}, we showed that the Schmidt number of all entangled states can indeed be certified by relaxing the requirement of full device independence in favor of a measurement-device-independent framework using semiquantum nonlocal games with trusted quantum inputs \cite{buscemi2012all} (see also \cite{shi2025schmidt}). While this approach overcomes constraint~$(i)$, constraints~$(ii)$ and~$(iii)$ remain unaddressed.

In this paper, we present a certification method that overcomes all the above constraints, while assuming trust only in the state-preparation devices and \textit{not} in the measurement devices. It is important to note that, while our certification method assumes two independent state-preparation devices, each faithfully preparing a given set of states, it does not impose the stronger requirement (constraint $(ii)$) that entangled states---such as maximally entangled states or states with full Schmidt rank---can be prepared reliably. Following \cite{PhysRevX.8.021033}, we refer to this as a setting with ``\textit{minimal assumptions}.''

In \cite{PhysRevX.8.021033}, Rosset \textit{et al.} introduced the framework of \textit{semiquantum signaling games} and showed that, for any quantum channel, it is possible to certify whether it is non-entanglement-breaking within the minimal-assumption setting using these games. Semiquantum games, in general, provide a useful operational framework for comparing quantum resources. For instance, the performance of shared quantum states in semiquantum nonlocal games \cite{buscemi2012all} provides a necessary and sufficient condition for determining whether one state can be converted into another by local operations and shared randomness (LOSR). Similarly, the convertibility of quantum channels under pre- and post-processing channels assisted by classical communication is determined by their performance in semiquantum signaling games \cite{PhysRevX.8.021033}. Since these games induce a resource-theoretic preorder among quantum states or channels, they can also be used to certify resourceful states or channels.

In the certification of non-entanglement-breaking channels in \cite{PhysRevX.8.021033}, a key element of the proof technique relies on a defining property of entanglement-breaking (EB) channels: they are equivalent to \textit{measure-and-prepare} channels. However, this characterization does not extend to Schmidt-number-breaking channels, preventing a direct generalization of their method. We overcome this limitation by using another structural property of Schmidt-number-breaking channels, formalized in Lemma~\ref{SNBC}. Using this property, we prove that any non-$k$-SNB channel can be certified via appropriately designed semiquantum signaling games. Interestingly, the same property also holds for another important class of channels, namely, nonpositive-partial-transpose-breaking (NPT-breaking) or positive under partial transpose (PPT) channels \cite{horodecki2000binding}, thereby allowing our method to be naturally extended to certify any non-NPT-breaking channel. Finally, we discuss a possible experimental implementation of our certification method. 

A related work \cite{Engineer2025certifying} has recently appeared on the semi-device-independent certification of non-$k$-SNB channels using quantum steering inequalities \cite{designolle2021genuine,de2023complete}. While that method overcomes constraints~$(ii)$ and~$(iii)$, it still suffers from constraint~$(i)$ due to the existence of entangled states whose entanglement, and a fortiori, Schmidt number cannot be certified in any steering-based scenario \cite{cavalcanti2016quantum, de2023complete}.

The rest of the paper is organized as follows. In Sec.~\ref{prelim}, we formally introduce Schmidt-number-breaking channels. Our main result is presented in Sec.~\ref{result}. In Sec.~\ref{implement}, we discuss the experimental implementation of our certification method through explicit examples. Finally, we conclude in Sec.~\ref{conclusion}.

\section{Schmidt-number-breaking Channels}\label{prelim}

Given a bipartite pure state $\ket{\psi}_{AB}\in\mathbb{C}^{d_A} \otimes \mathbb{C}^{d_B}$, it can be represented in the Schmidt decomposition form \cite{peres1997quantum, nielsen2010quantum} as
\begin{equation}
    \ket{\psi}_{AB} = \sum_{i=1}^{k} \sqrt{\lambda_i} \ket{u_i}_A \ket{v_i}_B,
\end{equation}
where \( \{\ket{u_i}_A\} \) and \( \{\ket{v_i}_B\} \) are orthonormal sets in \( \mathbb{C}^{d_A} \) and \( \mathbb{C}^{d_B} \), respectively, and $\sqrt{\lambda_i}$'s are the Schmidt coefficients, satisfying \( \lambda_i \geq 0 \) and \( \sum_{i} \lambda_i = 1 \). The number of non-zero Schmidt coefficients defines the \textit{Schmidt rank} (SR) of the state, with $1\leq \text{SR}(\ket{\psi}_{AB})\leq\min\{d_A,d_B\}$. $\text{SR}(\ket{\psi}_{AB})=1$ if and only if $\ket{\psi}_{AB}$ is a product state; otherwise, $\text{SR}(\ket{\psi}_{AB})$ quantifies the number of local degrees of freedom genuinely involved in the entanglement, i.e., the \textit{entanglement dimension}.

For a mixed bipartite state $\rho_{AB}\in\mathcal{D}(\mathbb{C}^{d_A} \otimes \mathbb{C}^{d_B})$, its entanglement dimension is quantified by its \textit{Schmidt number} (SN), defined as \cite{terhal2000schmidt}:
\begin{equation}
    \text{SN}(\rho_{AB}) = \min_{\{p_j, \ket{\psi_j}_{AB}\}} \left[ \max_j \, \text{SR}(\ket{\psi_j}_{AB}) \right],
\end{equation}
where the minimization is taken over all possible pure-state decomposition of $\rho_{AB}$, i.e., $\rho_{AB}=\sum_j p_j \ket{\psi_j}_{AB}\bra{\psi_j}$, with $\{p_j\}$ forming a probability vector and each $\ket{\psi_j}_{AB}\in\mathbb{C}^{d_A} \otimes \mathbb{C}^{d_B}$.

A quantum channel describes the most general evolution of a quantum state under interaction with an external environment. Mathematically, it is represented by a linear, completely positive (CP), and trace-preserving (TP) map from the space of linear operators on the input Hilbert space, $\mathcal{L}(\mathcal{H}_{\text{in}})$, to the space of linear operators on the output Hilbert space, $\mathcal{L}(\mathcal{H}_{\text{out}})$ \cite{nielsen2010quantum}. For simplicity, we consider $\mathcal{H}_{\text{in}} \cong \mathcal{H}_{\text{out}} \cong \mathcal{H}$, and refer to the map as acting on $\mathcal{L}(\mathcal{H})$; however, the analysis and results extend straightforwardly to the case where the input and output Hilbert spaces have different dimensions.

A channel $\mathcal{E}$ is called \textit{$k$-Schmidt-number-breaking} ($k$-SNB) if
\begin{equation}
\text{SN}\left[(\mathbf{id} \otimes \mathcal{E})(\rho_{AB})\right] \leq k, \quad \forall~\rho_{AB} \in \mathcal{D}(\mathbb{C}^{d_A} \otimes \mathbb{C}^{d_B}),
\end{equation}
where $\mathbf{id}$ is the identity channel on $\mathcal{L}(\mathbb{C}^{d_A})$ and $\mathcal{E}$ acts on $\mathcal{L}(\mathbb{C}^{d_B})$ \cite{chruscinski2006partially}. At first glance, this definition appears to require checking the action of $\mathcal{E}$ on infinitely many input states $\rho_{AB}$. However, Chruściński \textit{et al.}~\cite{chruscinski2006partially} provided an equivalent characterization of $k$-SNB channels in terms of the channel’s Choi-Jamiołkowski (CJ) operator.

According to the \textit{Choi-Jamiołkowski isomorphism} \cite{jamiolkowski1974effective, choi1975positive}, any channel $\mathcal{E}$ acting on $\mathcal{L}(\mathbb{C}^{d_B})$ can be uniquely represented by a bipartite operator $J_{BB'}^{\mathcal{E}} \in \mathcal{D}(\mathbb{C}^{d_B} \otimes \mathbb{C}^{d_{B'}})$, defined as
\begin{equation} \label{choistate}
J_{BB'}^{\mathcal{E}} = (\mathbf{id} \otimes \mathcal{E})(\ket{\Phi}_{BB'}\bra{\Phi}),
\end{equation}
where $\ket{\Phi}_{BB'} = \frac{1}{\sqrt{d_B}} \sum_{i=0}^{d_B-1} \ket{i}_B \ket{i}_{B'}$ is a maximally entangled state in $\mathbb{C}^{d_B} \otimes \mathbb{C}^{d_{B'}}$, and $\mathbb{C}^{d_B} \cong \mathbb{C}^{d_{B'}}$. Given the CJ operator $J^{\mathcal{E}}_{BB'}$, the action of the channel $\mathcal{E}$ on any state $\rho \in \mathcal{D}(\mathbb{C}^{d_B})$ is given by
\begin{equation}\label{inverse-choi}
    \mathcal{E}(\rho) = d_B \Tr_{B'} \left[ (\rho_{B'}^{\intercal} \otimes \mathbb{I}_B) \, J^{\mathcal{E}}_{BB'} \right],
\end{equation}
where $\rho^{\intercal}$ denotes the transpose of $\rho$ in the computational basis, and $\Tr_{B'}$ denotes the partial trace over subsystem $B'$.

\begin{lemma} \label{Choilemma}
    A channel $\mathcal{E}$ is a $k$-SNB channel if and only if $\text{SN}(J_{BB'}^{\mathcal{E}})\leq k$.
\end{lemma} 
The proof of this lemma can be found in \cite{chruscinski2006partially}. In particular, entanglement-breaking (EB) channels correspond to the special case $k = 1$.

Let us now define the set of all $k$-SNB channels as $k$-$\mathbb{SNBC}$. This set forms a convex and compact subset within the real vector space of all quantum channels and satisfies the following strict inclusion hierarchy:
\begin{small}
\begin{align}
1\text{-}\mathbb{SNBC} \subset 2\text{-}\mathbb{SNBC} \subset \cdots \subset k\text{-}\mathbb{SNBC}\subset\cdots d_{B}\text{-}\mathbb{SNBC}.
\end{align}
\end{small}

Due to the convexity and compactness of the $k\text{-}\mathbb{SNBC}$ set \cite{mallick2024characterization}, channels that are not $k$-SNB can be separated from this set using the Hahn-Banach separation theorem \cite{holmes2012geometric}. Specifically, for any channel $\mathcal{N}\notin k\text{-}\mathbb{SNBC}$, there exists a Hermitian operator $W$ (referred to as a witness operator) \cite{chruscinski2014entanglement} such that:
\begin{align}\label{witness}
    &\Tr \left[W J^{\mathcal{N}}_{BB'}\right]<0,\nonumber\\
    \text{and}\quad&\Tr \left[W J^{\mathcal{E}}_{BB'}\right]\geq0,\quad\forall~\mathcal{E}\in k\text{-}\mathbb{SNBC}. 
\end{align}
Note, however, that no single witness operator $W$ can, in general, detect all channels outside of $k$-$\mathbb{SNBC}$; that is, different non-$k$-SNB channels may require different witnesses.

Let us now move on to our method to certify that a given channel $\mathcal{N}$ is not $k$-SNB.

\section{Certification in the minimal-assumptions setting}\label{result}

We consider the framework of \textit{semiquantum signaling games}, introduced in \cite{PhysRevX.8.021033}. In this setting, a referee prepares a quantum system in a state chosen from a set of non-orthogonal states $\{\psi^x_A\}$ and sends it to a player, say Alice, who knows the identity of the set but not the specific state prepared. She then applies a quantum channel $\mathcal{N}$ to the system. This can correspond to two physically distinct but mathematically equivalent scenarios: if $\mathcal{N}$ represents a quantum memory, Alice stores the system in the memory, and the output state corresponds to the evolved state of the system retrieved from the memory at a later time in her lab; if $\mathcal{N}$ describes a communication link, Alice transmits the system through it to a different location, and the output state corresponds to the evolved state of the system in that location. In either case, we refer to the party in possession of the output state $\mathcal{N}(\psi^x_A)$ as Bob, even if it is Alice herself at a later time. The referee then sends another quantum system, prepared in a state from a known set of non-orthogonal states $\{\phi^y_B\}$, to Bob. Finally, Bob performs a joint measurement $M = \{\Pi^b_{AB}\}$ on the two systems and produces a classical output $b$. An illustration of this scenario is shown in Fig.~\ref{pic1}.

From the resulting temporal correlation
\begin{align} \label{prob}
    p(b|\psi_A^x, \phi_B^y) = \Tr\left[\left(\mathcal{N}(\psi_A^x) \otimes \phi_B^y\right)\Pi^b_{AB}\right],    
\end{align}
we can certify that $\mathcal{N}$ is not a $k$-SNB channel \textit{iff} no $k$-SNB channel $\mathcal{E}$ can reproduce this correlation for any choice of measurement $M$. To determine this, a suitable payoff function $\mathscr{J}(b,x,y)$ is assigned in each run of the game, based on the inputs $(\psi_A^x, \phi_B^y)$ and output $b$. The average payoff is then given by
\begin{align} \label{avgpayoff}
    \mathscr{J}_{\text{avg}}=\sum_{b,x,y}\mathscr{J}(b,x,y)~p(b|\psi_A^x, \phi_B^y)~p(\psi_A^x)~p(\phi_B^y),    
\end{align}
which must satisfy
\begin{align}\label{game_witness}
\left.\begin{aligned}
    &\mathscr{J}_{\text{avg}} \geq 0,~~\forall~\mathcal{E}\in k\text{-}\mathbb{SNBC},\\    
\text{and}\quad &\mathscr{J}_{\text{avg}} <0,~~\text{when}~\mathcal{N}\not\in k\text{-}\mathbb{SNBC}.
\end{aligned}\right\}
\end{align}
Here, $p(\psi_A^x)$ and $p(\phi_B^y)$ denote the probabilities with which the respective systems are prepared in the states $\psi_A^x$ and $\phi_B^y$.

\begin{figure}[t!]
\includegraphics[width=0.5\textwidth]{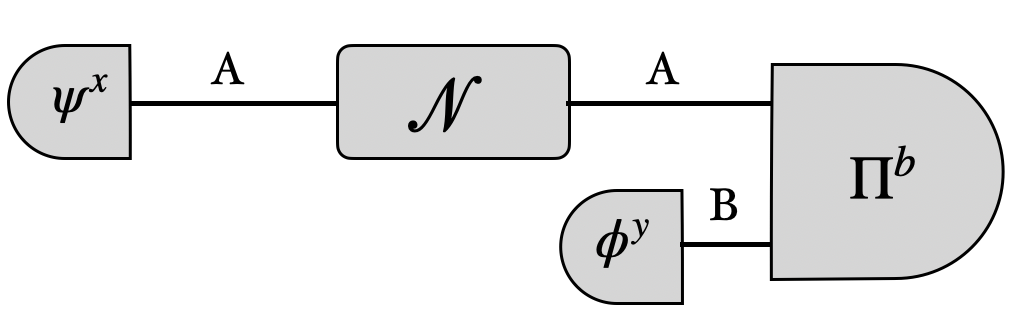} 
\caption{A schematic setup for the certification of non-$k$-SNB channels in a semiquantum signaling game with trusted quantum inputs.}\label{pic1}
\centering
\end{figure}

We now show that for every channel that is not $k$-SNB, one can construct a semiquantum signaling game with appropriate sets of states $\{\psi^x_A\}$, $\{\phi^y_B\}$, and a payoff function $\mathscr{J}(b,x,y)$ that certifies its non-$k$-SNB property. Before going to the main theorem, we first establish the following lemma.

\begin{lemma} \cite{mallick2024characterization} \label{SNBC}
    Let $\mathcal{E}$ be a $k$-SNB channel. Then, for any CP map $\mathcal{F}$, $\mathcal{F} \circ \mathcal{E}$ is a $k$-SNB CP map.
\end{lemma}

\begin{proof}
Consider the CJ operator $J^{\mathcal{F}\circ\mathcal{E}}_{AA'}$ of the composite map $\mathcal{F}\circ\mathcal{E}$, which can be written as
\begin{align*}
    J^{\mathcal{F}\circ\mathcal{E}}_{AA'}&=(\mathbf{id}\otimes\mathcal{F}\circ\mathcal{E})(\ket{\Phi}_{AA'}\bra{\Phi}),\\
    &=(\mathbf{id}\otimes\mathcal{F})\left[(\mathbf{id}\otimes\mathcal{E})(\ket{\Phi}_{AA'}\bra{\Phi})\right],\\
    &=(\mathbf{id}\otimes\mathcal{F})~J^{\mathcal{E}}_{AA'}
\end{align*}
where $J^{\mathcal{E}}_{AA'}$ is the CJ operator of channel $\mathcal{E}$. Now, since $\mathcal{E}$ is a $k$-SNB channel, we have $\text{SN}(J^{\mathcal{E}}_{AA'})\leq k$. This implies that $\text{SN}(J^{\mathcal{F}\circ\mathcal{E}}_{AA'})\leq k$, because the action of a local operation (CP map) $\mathcal{F}$ on $J^{\mathcal{E}}_{AA'}$ cannot increase its Schmidt number \cite{terhal2000schmidt}. Therefore, $\mathcal{F}\circ\mathcal{E}$ is a $k$-SNB CP map.
\end{proof}
\begin{theorem} \label{Theo1}
    For every channel $\mathcal{N} \notin k\text{-}\mathbb{SNBC}$, there exists a semiquantum signaling game for which the average payoff satisfies $\mathscr{J}_{\text{avg}}(\mathcal{N}) < 0$, while $\mathscr{J}_{\text{avg}}(\mathcal{E}) \geq 0$ for all $\mathcal{E} \in k\text{-}\mathbb{SNBC}$.
\end{theorem}

\begin{proof}
Recall that for all $\mathcal{N} \notin k\text{-}\mathbb{SNBC}$, there exists a witness operator $W$ satisfying Eq.~\eqref{witness}. Since the set of density operators spans the space of linear operators and $W$ is Hermitian, it can be expressed as  
\begin{equation}
    W_{AA'} = \sum_{x, y} \gamma_{x,y} \, (\xi^{x}_{A'} \otimes \zeta^{y}_{A}), \label{witnesss}
\end{equation}
where $\gamma_{x,y} \in \mathbb{R}$ for all $x, y$, and $\{\xi^{x}_{A'}\}$ and $\{\zeta^{y}_{A}\}$ are density operators on $\mathcal{H}_{A'}$ and $\mathcal{H}_A$, respectively, with $\mathcal{H}_A \cong \mathcal{H}_{A'}$. Note that this decomposition is, in general, non-unique.

Now consider a semiquantum signaling game defined by the input states $\psi_A^x = [\xi_A^x]^{\intercal}$ and $\phi_B^y = [\zeta_B^y]^{\intercal}$ for all $x, y$. Bob performs a two-outcome measurement \( M = \{\Pi^0_{AB}, \Pi^1_{AB}\} \), and the payoff function $\mathscr{J}(b,x,y)$ satisfies
\begin{align}\label{payoff}
    p(\psi_A^x)\, p(\phi_B^y)\, \mathscr{J}(b,x,y) =
    \begin{cases}
        \gamma_{x,y} & \text{if } b = 0, \\
        0 & \text{otherwise}.
    \end{cases}
\end{align}
We now show that this signaling game certifies that $\mathcal{N}$ is a non-$k$-SNB channel. For this the following conditions are to be satisfied:

\begin{enumerate}[label=(\roman*)]
    \item If $\mathcal{N} \notin k\text{-}\mathbb{SNBC}$, then there exists at least one measurement setting for which $\mathscr{J}_{\text{avg}} (\mathcal{N}) < 0$.
    \item For all measurement settings chosen by Bob, $\mathscr{J}_{\text{avg}} (\mathcal{E}) \ge 0$ for all $ \mathcal{E} \in k\text{-}\mathbb{SNBC}$.
\end{enumerate}
\textbf{Proof of the first condition:}  
    Let us consider that Bob performs the two-outcome joint measurement $\Pi^0_{AB}={\ket{\Phi} _{AB}\bra{\Phi}}$ and $\Pi^1_{AB}=\mathbb{I}_{AB} - \ket{\Phi}_{AB} \bra{\Phi}$ where $\ket{\Phi}_{AB} = \frac{1}{\sqrt{d}} \sum_{i=0}^{d-1} \ket{i}_A \ket{i}_B$, and $\mathcal{H}_A \cong \mathcal{H}_B \cong \mathbb{C}^d$. The average payoff then reads:
    \begin{align}
    \mathscr{J}_{\text{avg}}(\mathcal{N}) &= \sum_{x,y} \gamma_{x,y} \Tr_{AB}\left[ \left( \mathcal{N}([\xi^x_A]^{\intercal}) \otimes [\zeta^y_B]^{\intercal} \right) \ket{\Phi}_{AB}\bra{\Phi} \right] \nonumber \\
    \label{int}
\end{align}
    
From the Choi matrix $J^{\mathcal{N}}_{A'A}$, one can find the channel action $\mathcal{N}([\xi^{x}_{A}]^\intercal)$ from the inverse Choi isomorphism (given by Eq.~\eqref{inverse-choi}) as following:
\begin{equation} \label{inversechoi}
\begin{split}
    \mathcal{N}([\xi^{x}_{A}]^\intercal) =&  d_A \Tr_{A'}\left[J^{\mathcal{N}}_{A'A} \{{([\xi^{x}_{A'}]^\intercal)^\intercal} \otimes \mathbf{I}_A\}\right] \\
   =& d_A \Tr_{A'}\left[J^{\mathcal{N}}_{A'A} \{{[\xi^{x}_{A'}]} \otimes \mathbf{I}_A\}\right]  = {\mathcal{R}}^x_A 
\end{split}
\end{equation}
Substituting this in Eq.~\eqref{int}, we get

\begin{equation}
\begin{split}
  & \mathscr{J}_{\text{avg}}(\mathcal{N})\\
  &= \sum_{x,y} \gamma_{x,y} \Tr_{AB}\left[(\mathcal{R}^x_A \otimes [\zeta^{y}_{B}]^\intercal) \ket{\Phi}_{AB}\bra{\Phi}\right] \\
  &= \sum_{x,y} \gamma_{x,y}\Tr_{A} \left[\Tr_{B} \left\{ (\mathbb{I}_A \otimes [\zeta^{y}_{B}]^\intercal) \ket{\Phi}_{AB}\bra{\Phi} (\mathcal{R}^x_A \otimes \mathbb{I}_{B}) \right\}\right],
\end{split}
\label{modifypayoff}
\end{equation}

     Now again from Inverse Choi Isomorphism (Eq.~\eqref{inverse-choi}) we get
     \begin{equation}
     \begin{split}
         & \Tr_{B}\left[{} \{\mathbb{I}_A \otimes {([\zeta^{y}_{B}]^\intercal)}\} \ket{\Phi} _{AB}\bra{\Phi}\right]\\
         &= d_A \Tr_{B}\left[ \{\mathbf{I}_A \otimes {([\zeta^{y}_{B}]^\intercal)}\} J^{\mathbf{id}}_{AB}\right] \\
        & =\mathbf{id} (\zeta^{y}_{A})\\
         &=[\zeta^{y}_{A}]  
         \end{split}
         \label{identityinversechoi}
     \end{equation}
where, $J^{\mathbf{id}}_{AB}$ is the Choi operator corresponding to the identity map i.e., $J^{\mathbf{id}}_{AB} = {\ket{\Phi}}_{AB} \bra{\Phi} $.\\
Now substituting Eq.~\eqref{identityinversechoi} in Eq.~\eqref{modifypayoff}, we get 
\begin{equation}
\begin{split}
\mathscr{J}_{\text{avg}}(\mathcal{N})
     = & \sum_{x,y} \gamma_{x,y} \Tr_{A}\left[ [\zeta^{y}_{A}]\hspace{0.1cm} {\mathcal{R}}^x_A\right] \\
     = & d_A \sum_{x,y} \gamma_{x,y} \Tr_{A'A}\left[ (\xi^{x}_{A'} \otimes \zeta ^{y}_{A})\hspace{0.1cm} J^{\mathcal{N}}_{A'A}   \right] \hspace{0.5cm} \\
      = & d_A \Tr_{AA'}\left[ W_{AA'}J^{\mathcal{N}}_{AA'} \right] \\
       < & 0.
\end{split}
\label{pay}
\end{equation}
where the second and third lines are derived from Eqs.~\eqref{inversechoi} and \eqref{witnesss} respectively. 

\textbf{Proof of the second condition:}  
Bob performs the two-outcome measurement $M = \{\Pi^0_{AB}, \Pi^1_{AB}\}$. The average payoff, for this case, is given by
\begin{align}
    \mathscr{J}_{\text{avg}}(\mathcal{E})\nonumber
    &= \sum_{x,y} \gamma_{x,y} \Tr_{AB}\left[ \left( \mathcal{E}([\xi^x_A]^{\intercal}) \otimes [\zeta^y_B]^{\intercal} \right) \Pi^0_{AB} \right]\\
    &= \sum_{x,y} \gamma_{x,y}\Tr_{A} \left[\Tr_{B} \left\{ (\mathbb{I}_A \otimes [\zeta^{y}_{B}]^\intercal) \Pi
    ^{0}_{AB}\right\}\ (\mathcal{E} ([\xi^{x}_{A}]^{\intercal})\otimes \mathbb{I}_{B}) \right]\label{pospay}
\end{align}
Since $\Pi^0_{AB}$ is a positive operator, by the CJ isomorphism, we can consider $\Pi^0_{AB}$ as the CJ operator of a CP map $\mathcal{F}$ (not necessarily a channel), i.e., $\Pi^0_{AB} \equiv J^{\mathcal{F}}_{AB}$. Then, using Eq.~\eqref{inverse-choi}, we have
\begin{equation}
 d_A \Tr_B\left[ \left(\mathbb{I}_A \otimes [\zeta^{y}_{B}]^\intercal \right) \Pi^{0}_{AB}\right] =\mathcal{F} (\zeta^{y}_{A})   \label{inverseofpovm}
\end{equation}
Substituting this in Eq.~\eqref{pospay}, we get
\begin{align}
    \mathscr{J}_{\text{avg}}(\mathcal{E}) &= \sum_{x,y} \gamma_{x,y} \Tr_{A}\left[ \mathcal{E}([\xi^x_A]^{\intercal})\, \mathcal{F}(\zeta^y_A) \right] \nonumber \\
    &= \sum_{x,y} \gamma_{x,y} \Tr_{A}\left[ \left(\mathcal{F}^* \circ \mathcal{E} ([\xi^x_A]^{\intercal})\right)~\zeta^y_A \right], \label{inter}
\end{align}
where $\mathcal{F}^*$ is the adjoint map of $\mathcal{F}$. Since the adjoint of a CP map is also a CP map,  Lemma~\ref{SNBC} implies that $\mathcal{F}^* \circ \mathcal{E}$ is a $k$-SNB CP map. Using Eq.~\eqref{inverse-choi} again, we get
\begin{equation}
   \mathcal{F}^* \circ \mathcal{E} ([\xi^x_A]^{\intercal}) =  d_A \Tr_{A'}\left[ \left(\mathbb{I}_A \otimes \xi^{x}_{A'} \right) J^{\mathcal{F}^* \circ \mathcal{E}}_{AA'}\right] 
\end{equation}
Substituting this in Eq.~\eqref{inter}, the average payoff becomes
\begin{align}
    \mathscr{J}_{\text{avg}}(\mathcal{E}) &= d_A \sum_{x,y} \gamma_{x,y} \Tr_{AA'}\left[ (\xi^x_{A'} \otimes \zeta^y_A) J^{\mathcal{F}^* \circ \mathcal{E}}_{AA'} \right] \nonumber \\
    &= d_A \Tr_{AA'}\left[ W_{AA'} J^{\mathcal{F}^* \circ \mathcal{E}}_{AA'} \right] \nonumber\\
    & \geq 0,
\end{align}
where the last line follows from Eq.~\eqref{witnesss}. The final inequality follows from the properties of $W$ as a $k$-SNB witness and the fact that $\text{SN}(J^{\mathcal{F}^* \circ \mathcal{E}}_{AA'}) \leq k$.
\end{proof}

\begin{remark}
    Our certification method avoids the risk of falsely identifying a $k$-SNB channel as non-$k$-SNB due to measurement imperfections by exploiting the key property of $k$-SNB channels given in Lemma~\eqref{SNBC}. This property also applies to EB channels, since they correspond to the special case $k=1$. In \cite{PhysRevX.8.021033}, the same guarantee for EB channels was obtained via their alternative characterization as `measure-and-prepare' channels, which does not generalise to arbitrary $k$-SNB channels. Furthermore, our approach naturally extends to another generalization of EB channels, namely nonpositive-partial-transpose-breaking (NPT-breaking) or PPT channels \cite{horodecki2000binding}. The action of such channels on one half of any bipartite state yields an output state that is positive under partial transpose. These channels are efficiently characterized by their CJ operators: a channel is NPT-breaking iff its CJ operator is PPT. An analogue of Lemma~\ref{SNBC} also holds in this setting: if $\mathcal{P}$ is NPT-breaking, then for any CP map $\mathcal{F}$, the composition $\mathcal{F}\circ\mathcal{P}$ is an NPT-breaking CP map. This is discussed in detail in \ref{Appendix A}. 
    
    More generally, our method applies to any non-resource-breaking channel whose resource-breaking counterpart satisfies the analogues of Lemmas \ref{Choilemma} and \ref{SNBC}.
\end{remark}

 Note that our certification scheme is robust against particle loss and detector inefficiency (limitation $(iv)$ in the Introduction). In standard Bell tests, undetected events due to particle loss or detector inefficiency are often discarded. As a result, the observed correlations arise from a biased subset of detected events, which can lead to a false violation of Bell inequalities \cite{PhysRevD.2.1418,PhysRevA.46.3646,RevModPhys.86.419}. This limitation is often referred in the literature as the \textit{``detection loophole''}. In contrast, in our scheme undetected events may be assigned the outcome $b=1$. Since the payoff associated with this outcome is zero [Eq.~\eqref{payoff}], such events do not contribute to the violation of the inequalities in Eq.~\eqref{game_witness}. Importantly, all experimental rounds---including those corresponding to particle loss---are therefore included in the statistics. Consequently, particle loss can only reduce the observed violation and may lead to inconclusive results (false negatives), but it cannot increase the payoff beyond the bound satisfied by channels with bounded Schmidt number. Hence particle loss cannot lead to false certification in our protocol.

\section{Experimental implementation}\label{implement}

 Our method of certifying non-$k$-SNB channels is implementable in standard optical or light--matter hybrid platforms. In particular, recent experimental demonstrations of non-EB channel verification \cite{PhysRevLett.124.010503, PhysRevLett.127.160502} can, with suitable modifications, be adapted to realize our scheme.  

\begin{figure}[t!]
\includegraphics[width=0.7\textwidth]{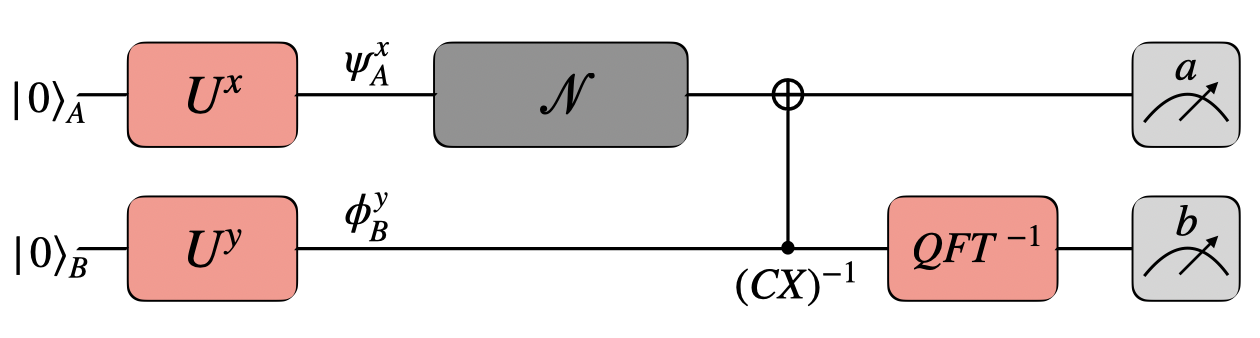} 
\caption{Circuit for experimental realization of the certification scheme for an unknown channel $\mathcal{N}$.}\label{fig1}
\centering
\end{figure}

In Fig.~\ref{fig1}, we illustrate the corresponding experimental circuit. Two families of unitary gates, $\{U_A^x\}$ and $\{U_B^y\}$, prepare the input states $\{\psi_A^x\}$ and $\{\phi_B^y\}$. For each run, values of $x$ and $y$ are chosen randomly, and the corresponding unitary gates are implemented. The state $\psi_A^x$ is then passed through the test channel $\mathcal{N}$, whose non-$k$-SNB property is to be certified. This is followed by an \textit{inverse controlled-shift} gate, where $\phi_B^y$ acts as the control and $\mathcal{N}(\psi_A^x)$ as the target. The controlled-shift operation is defined by 
\begin{align}
    CX \ket{i}\ket{j} = \ket{i \oplus_d j}\ket{j},    
\end{align}
where $i,j \in \{0,1,\dots,d-1\}$, with its inverse given by $(CX)^{-1} = (CX)^{\dagger}$. An \textit{inverse quantum Fourier transform} ($QFT^{-1}$) is then applied to the control, where the Fourier transform acts as 
\begin{align}
    QFT \ket{j} = \tfrac{1}{\sqrt{d}} \sum_{k=0}^{d-1} e^{i \frac{2\pi jk}{d}} \ket{k},~~\text{where}~j,k \in \{0,1,\dots,d-1\}.   
\end{align}
Finally, both systems are measured individually in the computational basis, producing the outcome pair $(a,b)$.

The sequence of the inverse controlled-shift gate and the inverse QFT gate, followed by measurement in the computational basis, effectively implements a measurement in the generalized Bell basis. The outcome pair $(a,b) = (0,0)$ corresponds to the measurement effect $\ket{\Phi}_{AB}\bra{\Phi}$, where $\ket{\Phi}_{AB}=\frac{1}{\sqrt{d}}\sum_{i=0}^{d-1}\ket{i}_A\ket{i}_B$. Therefore, to evaluate $\mathscr{J}_{\text{avg}}$, it suffices to collect the statistics of the outcome $(a,b)=(0,0)$ for different choices of $x$ and $y$.

It follows from the proof of Theorem \ref{Theo1} that the choice of the families of unitary gates and the parameter $\gamma_{x,y}$ depends on the specific Schmidt number witness operator employed to certify the channel. In what follows, we provide these specifications for an optimal Schmidt number witness operator, which, by definition, detects the maximal set of non-$k$-SNB channels \cite{sanpera2001schmidt}.\\  

\begin{example}
    Consider the depolarizing and dephasing channels on $\mathcal{L}(\mathbb{C}^3)$, whose actions are respectively given as
\begin{align}
    &\mathcal{D}_{\lambda}(\rho)=(1-\lambda)\rho+\lambda\Tr(\rho)\frac{\mathbb{I}_3}{3},\label{qutritdepol}\\
    \text{and}~~&\Phi_{\lambda}(\rho)=(1-\lambda)\rho+\lambda\sum_{i=0}^2 \ket{i}\bra{i}\rho\ket{i}\bra{i}\label{qutritdephase},
\end{align}
where $\lambda\in[0,1]$ is the noise parameter. We are interested in certifying the range of $\lambda$ for which they are not $2$-SNB. 
\end{example}

\begin{figure}[t!]
\includegraphics[width=1.1\textwidth]{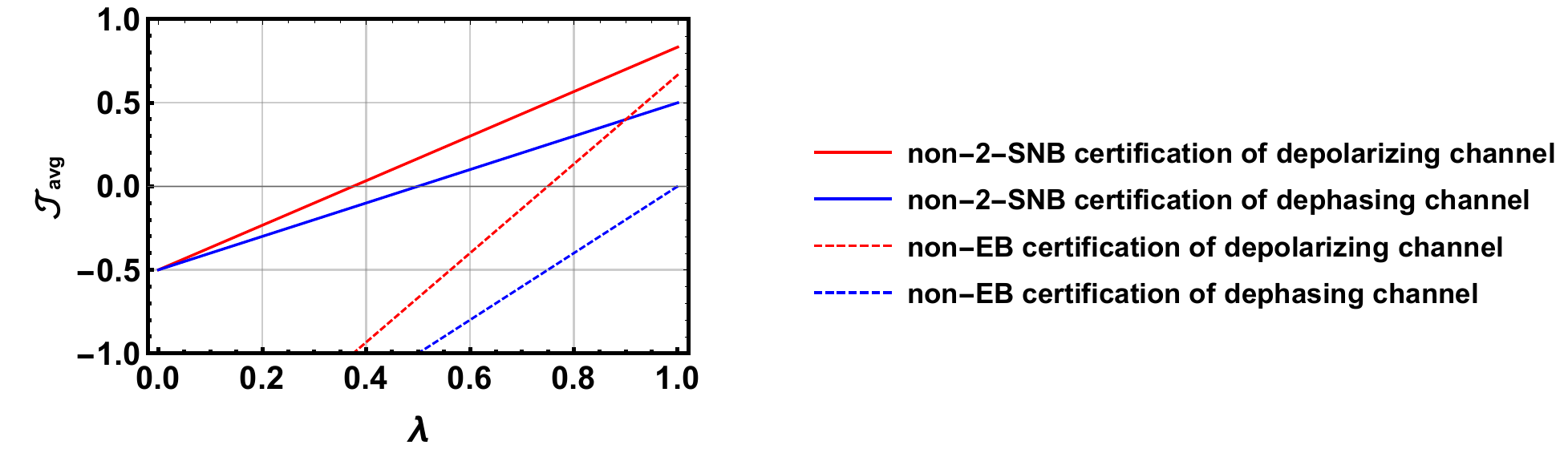} 
\caption{ Variation of the average payoff $\mathcal{J}_{\text{avg}}$ with respect to the noise parameter $\lambda$ for depolarizing and dephasing channels. The blue (red) curve corresponds to the dephasing (depolarizing) channel. A negative payoff in the solid (dotted) blue curve indicates the range of noise parameters for which the dephasing channel is non-$2$-SNB (non–EB). As evident from the plot, the dephasing channel is non-$2$-SNB for $0 \le \lambda < \frac{1}{2}$ and non-EB for $\lambda \in [0,1)$. As the noise parameter increases, the channel becomes $2$-SNB and eventually transitions to being entanglement-breaking (EB) at $\lambda=1$. A similar behavior is observed for the depolarizing channel, where the negative payoff in the solid and dotted red curves marks the non-$2$-SNB and non-EB regimes, respectively. In particular, the depolarizing channel is non-$2$-SNB for $0 \le \lambda < \frac{3}{8}$, and non-EB for $0 \le \lambda < \frac{3}{4}$. The average payoffs certifying the non-EB and non-$2$-SNB properties of the channels arise from different games with distinct payoff functions.}\label{plot2}
\centering
\end{figure}

To this end, we consider the following witness operator \cite{sanpera2001schmidt} detecting channels $\notin 2$-$\mathbb{SNBC}$:
\begin{align}\label{wit}
    W_2=\mathbb{I}_3\otimes\mathbb{I}_3-\frac{3}{2}\mathcal{P}_3,
\end{align}
where $\mathcal{P}_3$ denotes the projector onto the state $\ket{\Phi_3}=\frac{1}{\sqrt{3}}\sum_{i=0}^2\ket{ii}$. In particular, for depolarizing and dephasing channels, the negativity of the expectation value of $W_2$ is both a necessary and sufficient condition for the channels to be non-$2$-SNB. \( W_2 \) can be expressed as
\[
W_2 = \sum_{x,y=1}^9\gamma_{x, y}~(\xi^{x}\otimes\xi^{y}),
\]
where the coefficients \(\gamma_{x,y}\) are represented in the form of a symmetric matrix given by:
\begin{align}
    \{\gamma_{x,y}\}_{x,y=1}^9 \equiv\begin{pmatrix}
        \frac{1}{2} & 1 & 1 & \frac{1}{4} & \frac{1}{4} & 0 & -\frac{1}{4} & -\frac{1}{4} & 0 \\
        1 & \frac{1}{2} & 1 & \frac{1}{4} & 0 & \frac{1}{4} & -\frac{1}{4} & 0 & -\frac{1}{4}\\
        1 & 1 & \frac{1}{2} & 0 & \frac{1}{4} & \frac{1}{4} & 0 & -\frac{1}{4} & -\frac{1}{4}\\
        \frac{1}{4} & \frac{1}{4} & 0 & -\frac{1}{4} & 0 & 0 & 0 & 0 & 0\\
        \frac{1}{4} & 0 & \frac{1}{4} & 0 & -\frac{1}{4} & 0 & 0 & 0 & 0\\
        0 & \frac{1}{4} & \frac{1}{4} & 0 & 0 & -\frac{1}{4} & 0 & 0 & 0\\
        -\frac{1}{4} & -\frac{1}{4} & 0 & 0 & 0 & 0 & \frac{1}{4} & 0 & 0\\
        -\frac{1}{4} & 0 & -\frac{1}{4} & 0 & 0 & 0 & 0 & \frac{1}{4} & 0\\
        0 & -\frac{1}{4} & -\frac{1}{4} & 0 & 0 & 0 & 0 & 0 & \frac{1}{4}
    \end{pmatrix}
\end{align}

and $\xi^{x(y)}=\ketbra{\xi^{x(y)}}{\xi^{x(y)}}$'s are given by
\begin{align} \label{state}
    \begin{array}{ll}
            \ket{\xi^1}= \ket{0}, \quad & \quad \ket{\xi^{6}} = \frac{1}{\sqrt{2}}(\ket{1}+\ket{2}), \\
            \ket{\xi^{2}} = \ket{1}, \quad & \quad \ket{\xi^{7}} = \frac{1}{\sqrt{2}}(\ket{0} + i\ket{1}), \\
            \ket{\xi^{3}} = \ket{2}, \quad & \quad \ket{\xi^{8}} =  \frac{1}{\sqrt{2}}(\ket{0} + i\ket{2}), \\
            \ket{\xi^{4}} = \frac{1}{\sqrt{2}}(\ket{0}+\ket{1}), \quad & \quad \ket{\xi^{9}} =  \frac{1}{\sqrt{2}}(\ket{1} + i\ket{2}). \\
            \ket{\xi^{5}} = \frac{1}{\sqrt{2}}(\ket{0}+\ket{2}),
        \end{array}
\end{align}

The families of unitary gates are then given by
\begin{align*}
    \begin{array}{ll}
           U^1= \mathbb{I}_3, \quad & \quad U^6 =  \begin{pmatrix}
0 & 1 & 0\\
\frac{1}{\sqrt{2}} & 0 & \frac{1}{\sqrt{2}}\\
\frac{1}{\sqrt{2}} & 0 & \frac{-1}{\sqrt{2}}
\end{pmatrix} \vspace{0.3cm},\\
           U^2 =\begin{pmatrix}
0 & 1 & 0\\
1 & 0 & 0\\
0 & 0 & 1
\end{pmatrix}, \vspace{0.3cm} \quad & \quad U^7 = \begin{pmatrix}
\frac{1}{\sqrt{2}} & \frac{1}{\sqrt{2}} & 0\\
\frac{-i}{\sqrt{2}} & \frac{i}{\sqrt{2}} & 0\\
0 & 0 & 1
\end{pmatrix} \vspace{0.3cm},\\
            U^3 =\begin{pmatrix}
0 & 0 & 1\\
0 & 1 & 0\\
1 & 0 & 0
\end{pmatrix}, \vspace{0.3cm} \quad & \quad U^8 = \begin{pmatrix}
\frac{1}{\sqrt{2}} & 0 & \frac{1}{\sqrt{2}}\\
0 & 1 & 0\\
\frac{-i}{\sqrt{2}} & 0 & \frac{i}{\sqrt{2}}\\
\end{pmatrix} \vspace{0.3cm},  \\
            U^4 = \begin{pmatrix}
\frac{1}{\sqrt{2}} & \frac{1}{\sqrt{2}} & 0\\
\frac{1}{\sqrt{2}} & \frac{-1}{\sqrt{2}} & 0\\
0 & 0 & 1
\end{pmatrix} \vspace{0.3cm}, \quad & \quad U^9 = \begin{pmatrix}
0 & 1 & 0\\
\frac{1}{\sqrt{2}} & 0 & \frac{1}{\sqrt{2}}\\
\frac{-i}{\sqrt{2}} & 0 & \frac{i}{\sqrt{2}}\\
\end{pmatrix} \vspace{0.3cm}, \\
            U^5 = \begin{pmatrix}
\frac{1}{\sqrt{2}} & 0 & \frac{1}{\sqrt{2}}\\
0 & 1 & 0\\
\frac{1}{\sqrt{2}} & 0 & \frac{-1}{\sqrt{2}}
\end{pmatrix} \vspace{0.3cm},
        \end{array}
\end{align*}
such that $U^{x}\ket{0}=\ket{\psi^{x}}$ and $U^{y}\ket{0}=\ket{\phi^{y}}$, where $\psi^{x}=[\xi^{{x}}]^{\intercal}$ and $\phi^{y}=[\xi^{{y}}]^{\intercal}$.

Assigning the payoff function
\[\mathscr{J}(b,x,y) = \left\{
        \begin{array}{ll}
            81\gamma_{x,y} & \text{if } b=0, \\
            0 & \text{otherwise},
        \end{array}
    \right.\]
the average payoff $\mathscr{J}_{\text{avg}}$ obtained for the qutrit depolarizing and dephasing channels in Eqs. \eqref{qutritdepol} and \eqref{qutritdephase} are illustrated in Fig. \ref{plot2}. The negativity of $\mathscr{J}_{\text{avg}}$ certifies the noise-parameter range for which the channels are non-$2$-SNB.

 Note that one can similarly certify the range of $\lambda$ for which the channels are non-EB (also shown in Fig.~\ref{plot2}) by replacing the witness operator $W_2$ in Eq.~\eqref{wit} with an entanglement witness of the form \cite{sanpera2001schmidt}:
\begin{align}
W_1=\mathbb{I}_3\otimes\mathbb{I}_3-3\mathcal{P}_3.
\end{align}

\section{Conclusion}\label{conclusion}

We have presented a method to certify the non-$k$-Schmidt-number-breaking property of a given quantum channel, i.e., its ability to preserve entanglement of dimension exceeding $k$. Using a semiquantum game in the temporal setting, we have shown that non-$k$-SNB channels can generate input–output correlations that cannot be reproduced by any $k$-SNB channel under any choice of measurements, thereby enabling their certification. Our method differs from existing approaches, such as those based on Bell-nonlocal games or quantum steering, in that it is applicable to the certification of all non-$k$-SNB channels.

Although our method is not fully device-independent, it is measurement-device-independent and, importantly, eliminates the need for additional assumptions, such as the preparation of entangled input states or the availability of auxiliary noiseless side channels. The only assumption underlying our approach is that the devices preparing the input (product) quantum states for the semiquantum game are perfect. This assumption is more natural than trusting the measurement devices, which are inherently exposed to interactions with the external environment and therefore more susceptible to imperfections. While our certification method is robust against imperfections in the measurement devices, analyzing its robustness against imperfections or noise in the preparation devices remains an interesting direction for future work.

This work generalizes a previous framework of measurement-device-independent certification of non-entanglement-breaking channels to the broader family of non-$k$-SNB channels. This generalization relies on two key properties of Schmidt-number-breaking completely positive maps: $(i)$ their equivalent characterization in terms of their Choi–Jamiołkowski (CJ) operators, and $(ii)$ their closure under concatenation with arbitrary completely positive maps. Consequently, our method is sufficient to certify any non-resource-breaking channels whose resource-breaking counterparts satisfy these two conditions. Exploring whether channels that do not meet these conditions can also be certified within this minimal-assumption setting remains an intriguing open problem.

\begin{acknowledgements}
The authors acknowledge Manik Banik, Nirman Ganguly, Amit Kundu, and Sahil Gopalkrishna Naik for helpful discussions. B.M. acknowledges the DST INSPIRE fellowship program for financial support. P.G. acknowledges financial support from the project entitled “Technology Vertical - Quantum Communication” under the National Quantum Mission of the Department of Science and Technology (DST) (Sanction Order No. DST/QTC/NQM/QComm/2024/2 (G)).
\end{acknowledgements} 

\bibliography{main}

@article{sit2017high,
  title={High-dimensional intracity quantum cryptography with structured photons},
  author={Sit, Alicia and Bouchard, Fr{\'e}d{\'e}ric and Fickler, Robert and Gagnon-Bischoff, J{\'e}r{\'e}mie and Larocque, Hugo and Heshami, Khabat and Elser, Dominique and Peuntinger, Christian and G{\"u}nthner, Kevin and Heim, Bettina and others},
  journal={Optica},
  volume={4},
  number={9},
  pages={1006--1010},
  year={2017},
  publisher={Optical Society of America}
}

@Article{Dada2011,
author={Dada, Adetunmise C.
and Leach, Jonathan
and Buller, Gerald S.
and Padgett, Miles J.
and Andersson, Erika},
title={Experimental high-dimensional two-photon entanglement and violations of generalized Bell inequalities},
journal={Nature Physics},
year={2011},
month={Sep},
day={01},
volume={7},
number={9},
pages={677-680},
issn={1745-2481},
doi={10.1038/nphys1996},
url={https://doi.org/10.1038/nphys1996}
}

@article{Deutsch2020,
  title = {Harnessing the Power of the Second Quantum Revolution},
  author = {Deutsch, Ivan H.},
  journal = {PRX Quantum},
  volume = {1},
  issue = {2},
  pages = {020101},
  numpages = {13},
  year = {2020},
  month = {Nov},
  publisher = {American Physical Society},
  doi = {10.1103/PRXQuantum.1.020101},
  url = {https://link.aps.org/doi/10.1103/PRXQuantum.1.020101}
}

@article{Dowling2003,
  title={Quantum technology: the second quantum revolution},
  author={Dowling, Jonathan P and Milburn, Gerard J},
  journal={Philosophical Transactions of the Royal Society of London. Series A: Mathematical, Physical and Engineering Sciences},
  volume={361},
  number={1809},
  pages={1655--1674},
  year={2003},
  publisher={The Royal Society},
  doi = {10.1098/rsta.2003.1227},
  url = {https://doi.org/10.1098/rsta.2003.1227}
}

@article{Horodecki2003,
author = {Horodecki, Michael and Shor, Peter W. and Ruskai, Mary Beth},
title = {Entanglement Breaking Channels},
journal = {Reviews in Mathematical Physics},
volume = {15},
number = {06},
pages = {629-641},
year = {2003},
doi = {10.1142/S0129055X03001709},

URL = { 
    
        https://doi.org/10.1142/S0129055X03001709
    
    

},
eprint = { 
    
        https://doi.org/10.1142/S0129055X03001709
    
    

}
}

@Article{Malik2016,
author={Malik, Mehul
and Erhard, Manuel
and Huber, Marcus
and Krenn, Mario
and Fickler, Robert
and Zeilinger, Anton},
title={Multi-photon entanglement in high dimensions},
journal={Nature Photonics},
year={2016},
month={Apr},
day={01},
volume={10},
number={4},
pages={248-252},
issn={1749-4893},
doi={10.1038/nphoton.2016.12},
url={https://doi.org/10.1038/nphoton.2016.12}
}

@article{PhysRevX.8.021033,
  title = {Resource Theory of Quantum Memories and Their Faithful Verification with Minimal Assumptions},
  author = {Rosset, Denis and Buscemi, Francesco and Liang, Yeong-Cherng},
  journal = {Phys. Rev. X},
  volume = {8},
  issue = {2},
  pages = {021033},
  numpages = {15},
  year = {2018},
  month = {May},
  publisher = {American Physical Society},
  doi = {10.1103/PhysRevX.8.021033},
  url = {https://link.aps.org/doi/10.1103/PhysRevX.8.021033}
}

@article{chruscinski2006partially,
  title={On partially entanglement breaking channels},
  author={Chru{\'s}ci{\'n}ski, Dariusz and Kossakowski, Andrzej},
  journal={Open Systems \& Information Dynamics},
  volume={13},
  number={1},
  pages={17--26},
  year={2006},
  publisher={Springer},
  url={https://doi.org/10.1007/s11080-006-7264-7}
}

@article{RevModPhys.86.419,
  title = {Bell nonlocality},
  author = {Brunner, Nicolas and Cavalcanti, Daniel and Pironio, Stefano and Scarani, Valerio and Wehner, Stephanie},
  journal = {Rev. Mod. Phys.},
  volume = {86},
  issue = {2},
  pages = {419--478},
  numpages = {60},
  year = {2014},
  month = {Apr},
  publisher = {American Physical Society},
  doi = {10.1103/RevModPhys.86.419},
  url = {https://link.aps.org/doi/10.1103/RevModPhys.86.419}
}

@article{PhysicsPhysiqueFizika.1.195,
  title = {On the Einstein Podolsky Rosen paradox},
  author = {Bell, J. S.},
  journal = {Physics Physique Fizika},
  volume = {1},
  issue = {3},
  pages = {195--200},
  numpages = {6},
  year = {1964},
  month = {Nov},
  publisher = {American Physical Society},
  doi = {10.1103/PhysicsPhysiqueFizika.1.195},
  url = {https://link.aps.org/doi/10.1103/PhysicsPhysiqueFizika.1.195}
}

@article{PhysRevA.46.3646,
  title = {Critical analysis of the empirical tests of local hidden-variable theories},
  author = {Santos, Emilio},
  journal = {Phys. Rev. A},
  volume = {46},
  issue = {7},
  pages = {3646--3656},
  numpages = {0},
  year = {1992},
  month = {Oct},
  publisher = {American Physical Society},
  doi = {10.1103/PhysRevA.46.3646},
  url = {https://link.aps.org/doi/10.1103/PhysRevA.46.3646}
}

@article{sanpera2001schmidt,
  title={Schmidt-number witnesses and bound entanglement},
  author={Sanpera, Anna and Bru{\ss}, Dagmar and Lewenstein, Maciej},
  journal={Physical Review A},
  volume={63},
  number={5},
  pages={050301},
  year={2001},
  publisher={APS},
url={https://journals.aps.org/pra/abstract/10.1103/PhysRevA.63.050301}
}

@article{terhal2000schmidt,
  title={Schmidt number for density matrices},
  author={Terhal, Barbara M and Horodecki, Pawe{\l}},
  journal={Physical Review A},
  volume={61},
  number={4},
  pages={040301},
  year={2000},
  publisher={APS},
url={https://journals.aps.org/pra/abstract/10.1103/PhysRevA.61.040301}
}

@article{chruscinski2014entanglement,
  title={Entanglement witnesses: construction, analysis and classification},
  author={Chru{\'s}ci{\'n}ski, Dariusz and Sarbicki, Gniewomir},
  journal={Journal of Physics A: Mathematical and Theoretical},
  volume={47},
  number={48},
  pages={483001},
  year={2014},
  publisher={IOP Publishing},
url={https://iopscience.iop.org/article/10.1088/1751-8113/47/48/483001/meta}
}

@book{holmes2012geometric,
  title={Geometric functional analysis and its applications},
  author={Holmes, Richard B},
  volume={24},
  year={2012},
  publisher={Springer Science \& Business Media},
url={https://books.google.co.in/books?hl=en&lr=&id=_nnSBwAAQBAJ&oi=fnd&pg=PA1&dq=Geometric+functional+analysis+and+its+applications&ots=SetEoAy6o8&sig=_Gax7IYvTTwRDTHhjkjbHiMfO2g&redir_esc=y#v=onepage&q=Geometric%20functional%20analysis%20and%20its%20applications&f=false}
}

@article{buscemi2012all,
  title={All entangled quantum states are nonlocal},
  author={Buscemi, Francesco},
  journal={Physical review letters},
  volume={108},
  number={20},
  pages={200401},
  year={2012},
  publisher={APS},
url={https://journals.aps.org/prl/abstract/10.1103/PhysRevLett.108.200401}
}

@article{Engineer2025certifying,
  title = {Certifying high-dimensional quantum channels},
  author = {Engineer, Sophie and Goel, Suraj and Egelhaaf, Sophie and McCutcheon, Will and Srivastav, Vatshal and Leedumrongwatthanakun, Saroch and Wollmann, Sabine and Jones, Benjamin D. M. and Cope, Thomas and Brunner, Nicolas and Uola, Roope and Malik, Mehul},
  journal = {Phys. Rev. Res.},
  volume = {7},
  issue = {3},
  pages = {033233},
  numpages = {12},
  year = {2025},
  month = {Sep},
  publisher = {American Physical Society},
  doi = {10.1103/bwd5-wx7j},
  url = {https://link.aps.org/doi/10.1103/bwd5-wx7j}
}

@article{PhysRevLett.127.160502,
  title = {Measurement-Device-Independent Verification of a Quantum Memory},
  author = {Yu, Yong and Sun, Peng-Fei and Zhang, Yu-Zhe and Bai, Bing and Fang, Yu-Qiang and Luo, Xi-Yu and An, Zi-Ye and Li, Jun and Zhang, Jun and Xu, Feihu and Bao, Xiao-Hui and Pan, Jian-Wei},
  journal = {Phys. Rev. Lett.},
  volume = {127},
  issue = {16},
  pages = {160502},
  numpages = {6},
  year = {2021},
  month = {Oct},
  publisher = {American Physical Society},
  doi = {10.1103/PhysRevLett.127.160502},
  url = {https://link.aps.org/doi/10.1103/PhysRevLett.127.160502}
}

@article{PhysRevLett.124.010503,
  title = {Measurement-Device-Independent Verification of Quantum Channels},
  author = {Graffitti, Francesco and Pickston, Alexander and Barrow, Peter and Proietti, Massimiliano and Kundys, Dmytro and Rosset, Denis and Ringbauer, Martin and Fedrizzi, Alessandro},
  journal = {Phys. Rev. Lett.},
  volume = {124},
  issue = {1},
  pages = {010503},
  numpages = {6},
  year = {2020},
  month = {Jan},
  publisher = {American Physical Society},
  doi = {10.1103/PhysRevLett.124.010503},
  url = {https://link.aps.org/doi/10.1103/PhysRevLett.124.010503}
}

@article{bae2019more,
  title={More entanglement implies higher performance in channel discrimination tasks},
  author={Bae, Joonwoo and Chru{\'s}ci{\'n}ski, Dariusz and Piani, Marco},
  journal={Physical Review Letters},
  volume={122},
  number={14},
  pages={140404},
  year={2019},
  publisher={APS},
url={https://journals.aps.org/prl/abstract/10.1103/PhysRevLett.122.140404}
}

@article{cozzolino2019high,
  title={High-dimensional quantum communication: benefits, progress, and future challenges},
  author={Cozzolino, Daniele and Da Lio, Beatrice and Bacco, Davide and Oxenl{\o}we, Leif Katsuo},
  journal={Advanced Quantum Technologies},
  volume={2},
  number={12},
  pages={1900038},
  year={2019},
  publisher={Wiley Online Library},
url={https://onlinelibrary.wiley.com/doi/full/10.1002/qute.201900038}
}

@article{lanyon2009simplifying,
  title={Simplifying quantum logic using higher-dimensional Hilbert spaces},
  author={Lanyon, Benjamin P and Barbieri, Marco and Almeida, Marcelo P and Jennewein, Thomas and Ralph, Timothy C and Resch, Kevin J and Pryde, Geoff J and O’brien, Jeremy L and Gilchrist, Alexei and White, Andrew G},
  journal={Nature Physics},
  volume={5},
  number={2},
  pages={134--140},
  year={2009},
  publisher={Nature Publishing Group UK London},
  url={https://doi.org/10.1038/nphys1150}
}

@article{PhysRevA.71.044305,
  title = {Quantum secure direct communication with high-dimension quantum superdense coding},
  author = {Wang, Chuan and Deng, Fu-Guo and Li, Yan-Song and Liu, Xiao-Shu and Long, Gui Lu},
  journal = {Phys. Rev. A},
  volume = {71},
  issue = {4},
  pages = {044305},
  numpages = {4},
  year = {2005},
  month = {Apr},
  publisher = {American Physical Society},
  doi = {10.1103/PhysRevA.71.044305},
  url = {https://link.aps.org/doi/10.1103/PhysRevA.71.044305}
}

@article{Mirhosseini_2015,
doi = {10.1088/1367-2630/17/3/033033},
url = {https://dx.doi.org/10.1088/1367-2630/17/3/033033},
year = {2015},
month = {mar},
publisher = {IOP Publishing},
volume = {17},
number = {3},
pages = {033033},
author = {Mirhosseini, Mohammad and Magaña-Loaiza, Omar S and O’Sullivan, Malcolm N and Rodenburg, Brandon and Malik, Mehul and Lavery, Martin P J and Padgett, Miles J and Gauthier, Daniel J and Boyd, Robert W},
title = {High-dimensional quantum cryptography with twisted light},
journal = {New Journal of Physics}
}

@article{PhysRevA.88.032309,
  title = {Weak randomness in device-independent quantum key distribution and the advantage of using high-dimensional entanglement},
  author = {Huber, Marcus and Paw\l{}owski, Marcin},
  journal = {Phys. Rev. A},
  volume = {88},
  issue = {3},
  pages = {032309},
  numpages = {7},
  year = {2013},
  month = {Sep},
  publisher = {American Physical Society},
  doi = {10.1103/PhysRevA.88.032309},
  url = {https://link.aps.org/doi/10.1103/PhysRevA.88.032309}
}

@article{hirsch2020schmidt,
  title={The Schmidt number of a quantum state cannot always be device-independently certified},
  author={Hirsch, Flavien and Huber, Marcus},
  journal={arXiv preprint arXiv:2003.14189},
  year={2020},
url={https://arxiv.org/abs/2003.14189#:~:text=In%20this%20paper%20we%20put,on%20any%20number%20of%20copies.}
}

@book{peres1997quantum,
  title={Quantum theory: concepts and methods},
  author={Peres, Asher},
  volume={72},
  year={1997},
  publisher={Springer},
url={https://link.springer.com/book/10.1007/0-306-47120-5}
}

@book{nielsen2010quantum,
  title={Quantum computation and quantum information},
  author={Nielsen, Michael A and Chuang, Isaac L},
  year={2010},
  publisher={Cambridge university press},
url={https://books.google.co.in/books?hl=en&lr=&id=-s4DEy7o-a0C&oi=fnd&pg=PR17&dq=M.A.+Nielsen+and+I.L.+Chuang,+Quantum+Computation+and++Quantum+Information+(Cambridge+University+Press,++Cambridge,+England,+2010)&ots=NJ4Jjrsu0s&sig=0ETLYT7gCfDRlDFY0s7tAr-CCj8&redir_esc=y#v=onepage&q&f=false}
}

@article{mallick2024characterization,
  title={On the characterization of Schmidt number breaking and annihilating channels},
  author={Mallick, Bivas and Ganguly, Nirman and Majumdar, A S},
  journal={arXiv preprint arXiv:2411.19315},
  year={2024},
url={https://arxiv.org/abs/2411.19315}
}

@article{jamiolkowski1974effective,
  title={An effective method of investigation of positive maps on the set of positive definite operators},
  author={Jamio{\l}kowski, A},
  journal={Reports on Mathematical Physics},
  volume={5},
  number={3},
  pages={415--424},
  year={1974},
  publisher={Elsevier},
url={https://www.sciencedirect.com/science/article/abs/pii/0034487774900445},
}

@article{choi1975positive,
  title={Positive semidefinite biquadratic forms},
  author={Choi, Man-Duen},
  journal={Linear Algebra and its applications},
  volume={12},
  number={2},
  pages={95--100},
  year={1975},
  publisher={North-Holland},
url={https://core.ac.uk/download/pdf/82119059.pdf},
}

@article{horodecki2000binding,
  title={Binding entanglement channels},
  author={Horodecki, Pawe{\l} and Horodecki, Micha{\l} and Horodecki, Ryszard},
  journal={Journal of Modern Optics},
  volume={47},
  number={2-3},
  pages={347--354},
  year={2000},
  publisher={Taylor \& Francis},
url={https://www.tandfonline.com/doi/abs/10.1080/09500340008244047}
}

@article{pal2015non,
  title={Non-locality breaking qubit channels: the case for CHSH inequality},
  author={Pal, Rajarshi and Ghosh, Sibasish},
  journal={Journal of Physics A: Mathematical and Theoretical \textbf{48}, 155302 (2015)},
 url={https://iopscience.iop.org/article/10.1088/1751-8113/48/15/155302/meta}
}

@article{srinidhi2024quantum,
  title={Quantum channels that destroy negative conditional entropy},
  author={Srinidhi, PV and Chakrabarty, Indranil and Bhattacharya, Samyadeb and Ganguly, Nirman},
  journal={Physical Review A \textbf{110}, 042423 (2024)},
url={https://journals.aps.org/pra/abstract/10.1103/PhysRevA.110.042423}
}

@article{luo2022coherence,
  title={Coherence-breaking superchannels},
  author={Luo, Yu and Li, Yongming and Xi, Zhengjun},
  journal={Quantum Information Processing \textbf{21}, 176 (2022)},
 url={https://link.springer.com/article/10.1007/s11128-022-03511-y}
}

@article{muhuri2023information,
  title={Information theoretic resource-breaking channels},
  author={Muhuri, Abhishek and Patra, Ayan and Gupta, Rivu and De, Aditi Sen},
  journal={arXiv preprint arXiv:2309.03108},
  year={(2023)},
url={https://arxiv.org/abs/2309.03108}
}

@article{heinosaari2015incompatibility,
  title={Incompatibility breaking quantum channels},
  author={Heinosaari, Teiko and Kiukas, Jukka and Reitzner, Daniel and Schultz, Jussi},
  journal={Journal of Physics A: Mathematical and Theoretical \textbf{48}, 435301 (2015)},
url={https://iopscience.iop.org/article/10.1088/1751-8113/48/43/435301/meta}
}

@article{ku2022quantifying,
  title={Quantifying quantumness of channels without entanglement},
  author={Ku, Huan-Yu and Kadlec, Josef and {\v{C}}ernoch, Anton{\'\i}n and Quintino, Marco T{\'u}lio and Zhou, Wenbin and Lemr, Karel and Lambert, Neill and Miranowicz, Adam and Chen, Shin-Liang and Nori, Franco and others},
  journal={PRX Quantum \textbf{3}, 020338 (2022)},
url={https://journals.aps.org/prxquantum/abstract/10.1103/PRXQuantum.3.020338}
}

@article{designolle2021genuine,
  title={Genuine high-dimensional quantum steering},
  author={Designolle, S{\'e}bastien and Srivastav, Vatshal and Uola, Roope and Valencia, Natalia Herrera and McCutcheon, Will and Malik, Mehul and Brunner, Nicolas},
  journal={Physical review letters},
  volume={126},
  number={20},
  pages={200404},
  year={2021},
  publisher={APS},
url={https://journals.aps.org/prl/abstract/10.1103/PhysRevLett.126.200404}
}

@article{de2023complete,
  title={Complete hierarchy for high-dimensional steering certification},
  author={de Gois, Carlos and Pl{\'a}vala, Martin and Schwonnek, Ren{\'e} and G{\"u}hne, Otfried},
  journal={Physical Review Letters},
  volume={131},
  number={1},
  pages={010201},
  year={2023},
  publisher={APS},
url={https://journals.aps.org/prl/abstract/10.1103/PhysRevLett.131.010201}
}

@article{cavalcanti2016quantum,
  title={Quantum steering: a review with focus on semidefinite programming},
  author={Cavalcanti, Daniel and Skrzypczyk, Paul},
  journal={Reports on Progress in Physics},
  volume={80},
  number={2},
  pages={024001},
  year={2016},
  publisher={IOP Publishing},
url={https://iopscience.iop.org/article/10.1088/1361-6633/80/2/024001/meta}
}

@article{kumar2025fidelity,
  title={Fidelity of entanglement and quantum entropies: unveiling their relationship in quantum states and channels: K. Kumar et al.},
  author={Kumar, Komal and Mallick, Bivas and Patro, Tapaswini and Ganguly, Nirman},
  journal={The European Physical Journal Plus},
  volume={140},
  number={10},
  pages={989},
  year={2025},
  publisher={Springer},
url={https://link.springer.com/article/10.1140/epjp/s13360-025-06941-6}
}

@article{patra2024qubit,
  title={Qubit magic-breaking channels},
  author={Patra, Ayan and Gupta, Rivu and Ferraro, Alessandro and De, Aditi Sen},
  journal={arXiv preprint arXiv:2409.04425},
  year={2024},
url={https://arxiv.org/abs/2409.04425}
}

@article{devendra2023mapping,
  title={Mapping cone of k-entanglement breaking maps},
  author={Devendra, Repana and Mallick, Nirupama and Sumesh, Kappil},
  journal={Positivity},
  volume={27},
  number={1},
  pages={5},
  year={2023},
  publisher={Springer},
url={https://link.springer.com/article/10.1007/s11117-022-00956-4}
}

@article{mallick2025higher,
  title={Higher-dimensional-entanglement detection and quantum-channel characterization using moments of generalized positive maps},
  author={Mallick, Bivas and Maity, Ananda G and Ganguly, Nirman and Majumdar, AS},
  journal={Physical Review A},
  volume={112},
  number={1},
  pages={012416},
  year={2025},
  publisher={APS},
url={https://journals.aps.org/pra/abstract/10.1103/nzrc-8yrt}
}

@article{johnston2008partially,
  title={Partially entanglement breaking maps and right CP-invariant cones},
  author={Johnston, Nathaniel},
  journal={Unpublished preprint},
  year={2008},
url={https://www.njohnston.ca/wp-content/uploads/2008/12/PEB_2009.pdf}
}

@article{shirokov2013schmidt,
  title={Schmidt number and partially entanglement-breaking channels in infinite-dimensional quantum systems},
  author={Shirokov, Maksim Evgenievich},
  journal={Mathematical Notes},
  volume={93},
  number={5},
  pages={766--779},
  year={2013},
  publisher={Springer},
url={https://link.springer.com/article/10.1134/S0001434613050143}
}

@article{groblacher2006experimental,
  title={Experimental quantum cryptography with qutrits},
  author={Gr{\"o}blacher, Simon and Jennewein, Thomas and Vaziri, Alipasha and Weihs, Gregor and Zeilinger, Anton},
  journal={New Journal of Physics},
  volume={8},
  number={5},
  pages={75},
  year={2006},
  publisher={IOP Publishing},
url={https://iopscience.iop.org/article/10.1088/1367-2630/8/5/075/meta}
}

@article{nlz1-h6qr,
  title = {Measurement-device-independent Schmidt number certification of all entangled states},
  author = {Mukherjee, Saheli and Mallick, Bivas and Das, Arun Kumar and Kundu, Amit and Ghosal, Pratik},
  journal = {Phys. Rev. A},
  volume = {112},
  issue = {6},
  pages = {062434},
  numpages = {11},
  year = {2025},
  month = {Dec},
  publisher = {American Physical Society},
  doi = {10.1103/nlz1-h6qr},
  url = {https://link.aps.org/doi/10.1103/nlz1-h6qr}
}

@article{PhysRevLett.76.722,
  title = {Purification of Noisy Entanglement and Faithful Teleportation via Noisy Channels},
  author = {Bennett, Charles H. and Brassard, Gilles and Popescu, Sandu and Schumacher, Benjamin and Smolin, John A. and Wootters, William K.},
  journal = {Phys. Rev. Lett.},
  volume = {76},
  issue = {5},
  pages = {722--725},
  numpages = {0},
  year = {1996},
  month = {Jan},
  publisher = {American Physical Society},
  doi = {10.1103/PhysRevLett.76.722},
  url = {https://link.aps.org/doi/10.1103/PhysRevLett.76.722}
}

@article{PhysRevLett.80.5239,
  title = {Mixed-State Entanglement and Distillation: Is there a ``Bound'' Entanglement in Nature?},
  author = {Horodecki, Micha\l{} and Horodecki, Pawe\l{} and Horodecki, Ryszard},
  journal = {Phys. Rev. Lett.},
  volume = {80},
  issue = {24},
  pages = {5239--5242},
  numpages = {0},
  year = {1998},
  month = {Jun},
  publisher = {American Physical Society},
  doi = {10.1103/PhysRevLett.80.5239},
  url = {https://link.aps.org/doi/10.1103/PhysRevLett.80.5239}
}

@article{PhysRevA.60.1888,
  title = {General teleportation channel, singlet fraction, and quasidistillation},
  author = {Horodecki, Micha\l{} and Horodecki, Pawe\l{} and Horodecki, Ryszard},
  journal = {Phys. Rev. A},
  volume = {60},
  issue = {3},
  pages = {1888--1898},
  numpages = {0},
  year = {1999},
  month = {Sep},
  publisher = {American Physical Society},
  doi = {10.1103/PhysRevA.60.1888},
  url = {https://link.aps.org/doi/10.1103/PhysRevA.60.1888}
}

@article{PhysRevLett.102.250501,
  title = {All Entangled States are Useful for Channel Discrimination},
  author = {Piani, Marco and Watrous, John},
  journal = {Phys. Rev. Lett.},
  volume = {102},
  issue = {25},
  pages = {250501},
  numpages = {4},
  year = {2009},
  month = {Jun},
  publisher = {American Physical Society},
  doi = {10.1103/PhysRevLett.102.250501},
  url = {https://link.aps.org/doi/10.1103/PhysRevLett.102.250501}
}

@article{PhysRevX.5.041011,
  title = {Extractable Work from Correlations},
  author = {Perarnau-Llobet, Mart\'{\i} and Hovhannisyan, Karen V. and Huber, Marcus and Skrzypczyk, Paul and Brunner, Nicolas and Ac\'{\i}n, Antonio},
  journal = {Phys. Rev. X},
  volume = {5},
  issue = {4},
  pages = {041011},
  numpages = {14},
  year = {2015},
  month = {Oct},
  publisher = {American Physical Society},
  doi = {10.1103/PhysRevX.5.041011},
  url = {https://link.aps.org/doi/10.1103/PhysRevX.5.041011}
}

@article{PhysRevA.103.062422,
  title = {Entanglement of a bipartite channel},
  author = {Gour, Gilad and Scandolo, Carlo Maria},
  journal = {Phys. Rev. A},
  volume = {103},
  issue = {6},
  pages = {062422},
  numpages = {21},
  year = {2021},
  month = {Jun},
  publisher = {American Physical Society},
  doi = {10.1103/PhysRevA.103.062422},
  url = {https://link.aps.org/doi/10.1103/PhysRevA.103.062422}
}

@article{szymanski2017convex,
  title={Convex set of quantum states with positive partial transpose analysed by hit and run algorithm},
  author={Szyma{\'n}ski, Konrad and Collins, Benot and Szarek, Tomasz and {\.Z}yczkowski, Karol},
  journal={Journal of Physics A: Mathematical and Theoretical},
  volume={50},
  number={25},
  pages={255206},
  year={2017},
  publisher={IOP Publishing},
url={https://iopscience.iop.org/article/10.1088/1751-8121/aa70f5/meta}
}

@article{PhysRevLett.82.1056,
  title = {Bound Entanglement Can Be Activated},
  author = {Horodecki, Pawe\l{} and Horodecki, Micha\l{} and Horodecki, Ryszard},
  journal = {Phys. Rev. Lett.},
  volume = {82},
  issue = {5},
  pages = {1056--1059},
  numpages = {0},
  year = {1999},
  month = {Feb},
  publisher = {American Physical Society},
  doi = {10.1103/PhysRevLett.82.1056},
  url = {https://link.aps.org/doi/10.1103/PhysRevLett.82.1056}
}

@article{PhysRevA.59.1070,
  title = {Quantum nonlocality without entanglement},
  author = {Bennett, Charles H. and DiVincenzo, David P. and Fuchs, Christopher A. and Mor, Tal and Rains, Eric and Shor, Peter W. and Smolin, John A. and Wootters, William K.},
  journal = {Phys. Rev. A},
  volume = {59},
  issue = {2},
  pages = {1070--1091},
  numpages = {0},
  year = {1999},
  month = {Feb},
  publisher = {American Physical Society},
  doi = {10.1103/PhysRevA.59.1070},
  url = {https://link.aps.org/doi/10.1103/PhysRevA.59.1070}
}

@article{PhysRevLett.82.5385,
  title = {Unextendible Product Bases and Bound Entanglement},
  author = {Bennett, Charles H. and DiVincenzo, David P. and Mor, Tal and Shor, Peter W. and Smolin, John A. and Terhal, Barbara M.},
  journal = {Phys. Rev. Lett.},
  volume = {82},
  issue = {26},
  pages = {5385--5388},
  numpages = {0},
  year = {1999},
  month = {Jun},
  publisher = {American Physical Society},
  doi = {10.1103/PhysRevLett.82.5385},
  url = {https://link.aps.org/doi/10.1103/PhysRevLett.82.5385}
}

@article{shi2024families,
  title={Families of Schmidt-number witnesses for high dimensional quantum states},
  author={Shi, Xian},
  journal={Communications in Theoretical Physics},
  volume={76},
  number={8},
  pages={085103},
  year={2024},
  publisher={IOP Publishing},
url={https://iopscience.iop.org/article/10.1088/1572-9494/ad48fb/meta}
}

@article{PhysRevD.2.1418,
  title = {Hidden-Variable Example Based upon Data Rejection},
  author = {Pearle, Philip M.},
  journal = {Phys. Rev. D},
  volume = {2},
  issue = {8},
  pages = {1418--1425},
  numpages = {0},
  year = {1970},
  month = {Oct},
  publisher = {American Physical Society},
  doi = {10.1103/PhysRevD.2.1418},
  url = {https://link.aps.org/doi/10.1103/PhysRevD.2.1418}}

@article{shi2025schmidt,
  title={Schmidt-number robustness as a unified quantifier of high dimensional entanglement in Buscemi nonlocality},
  author={Shi, Xian},
  journal={arXiv preprint arXiv:2506.07195},
  year={2025},
  url={https://arxiv.org/abs/2506.07195}
}

\onecolumngrid
\appendix

\section{Certification of all non-NPT-breaking channels with minimal assumptions}\label{Appendix A}
While we have explicitly discussed the certification of non-$k$-Schmidt number breaking channels in the minimal assumption scenario, our formalism can be applied to certify any non-resource-breaking channel that satisfies the analogues of Lemmas \ref{Choilemma} and \ref{SNBC}. Here, we show how our certification scheme can be applied to certify non-NPT-breaking channels by providing an explicit example. To begin with, we begin by briefly reviewing the definition and basic properties of NPT states and NPT-breaking channels.

\begin{definition}
    A bipartite state $\rho_{AB}$ is said to nonpositive under partial transpose (NPT) if
    \begin{align}
        \rho_{AB}^{\intercal_{B}}\not\geq 0.
    \end{align}
\end{definition}

It is well known that a state can be NPT only if it is entangled \cite{nielsen2010quantum}. However, the converse does not hold in general: there exist entangled states $\sigma_{AB}$ that are positive under partial transpose (PPT), i.e., $\sigma^{\intercal_{B}}_{AB}\geq 0$ \cite{PhysRevA.59.1070,PhysRevLett.82.5385}. NPT entangled states are particularly valuable resources for several quantum information processing tasks, including entanglement distillation \cite{PhysRevLett.76.722,PhysRevLett.80.5239}, quantum teleportation \cite{PhysRevA.60.1888}, channel discrimination \cite{PhysRevLett.102.250501}, work extraction \cite{PhysRevX.5.041011}, among others. The use of NPT states in the context of quantum channels has also been explored \cite{PhysRevA.103.062422}. Consequently, identifying and certifying channels that preserve the NPT character of entangled states is of fundamental and operational importance.

\begin{definition} \label{def1}
    A channel $\mathcal{S}: \mathcal{L}(\mathbb{C}^{d_B})\rightarrow \mathcal{L}(\mathbb{C}^{d_B})$ is called an NPT-breaking channel if
    \begin{equation*}
        \left[(\mathbf{id} \otimes \mathcal{S})(\rho_{AB})\right] \in \mathbb{PPT}, \quad \forall~\rho_{AB} \in \mathcal{D}(\mathbb{C}^{d_A} \otimes \mathbb{C}^{d_B})
    \end{equation*}
    where $\mathbb{PPT}\subset \mathcal{D}(\mathbb{C}^{d_A}\otimes\mathbb{C}^{d_B})$ represents the set of all states that remain positive under partial transpose (PPT).
\end{definition}
An equivalent characterization of NPT-breaking channels, analogous to that of k-SNB channels [Lemma \ref{Choilemma}], is given by the following lemma \cite{horodecki2000binding}.
\begin{lemma} \label{lemma3}
    A channel $\mathcal{S}$ is NPT-breaking if and only if 
    \begin{equation*}
        \left[(\mathbf{id} \otimes \mathcal{S}) (J_{BB'}^{\mathcal{S}})\right] \in \mathbb{PPT}\subset \mathcal{D}(\mathbb{C}^{d_B}\otimes\mathbb{C}^{d_{B'}}).
    \end{equation*}
    where $J_{BB'}^{\mathcal{S}}$ is the CJ operator of the channel $\mathcal{S}$, as defined in Eq.~\eqref{choistate}.
\end{lemma}

The proof of this Lemma can be found in \cite{horodecki2000binding}.\\

From definition \ref{def1} and the convexity and compactness of the set $\mathbb{PPT}$ \cite{szymanski2017convex}, it follows that the set of all NPT-breaking channels, denoted by $\mathbb{NPTBC}$, forms a convex and compact subset of the set of all quantum channels on $\mathcal{L}(\mathbb{C}^{d_B})$. Therefore, by the Hahn-Banach separation theorem \cite{holmes2012geometric}, for any non-NPT-breaking channel $\mathcal{R}$, there exists a witness operator $\tilde{W}$, satisfying
\begin{align}
    &\Tr(\tilde{W} J_{BB'}^{\mathcal{R}}) <
    0,\\
    \text{and}\quad &\Tr(\tilde{W} J_{BB'}^{\mathcal{S}}) \ge 0, \quad \forall \mathcal{S} \in \mathbb{NPTBC},
\end{align}
that separates it from the set of NPT-breaking channels.

Furthermore, NPT-breaking channels also obey a property analogous to Lemma \ref{SNBC}:
\begin{lemma} \label{lemma4}
    Let $\mathcal{S}$ be an NPT-breaking channel. Then, for any CP map $\mathcal{F}$, $\mathcal{F} \circ \mathcal{S}$ is an NPT-breaking CP map. 
\end{lemma}
\begin{proof}
    From lemma \ref{lemma3}, it suffices to prove that the CJ operator of the composite map $\mathcal{F} \circ \mathcal{S}$ is a PPT operator. Consider the CJ operator ($J^{\mathcal{F}\circ\mathcal{S}}_{AA'}$) of the composite map $\mathcal{F}\circ\mathcal{S}$, which can be written as
\begin{align*}
    J^{\mathcal{F}\circ\mathcal{S}}_{AA'}&=(\mathbf{id}\otimes\mathcal{F}\circ\mathcal{S})(\ket{\Phi}_{AA'}\bra{\Phi}),\\
    &=(\mathbf{id}\otimes\mathcal{F})\left[(\mathbf{id}\otimes\mathcal{S})(\ket{\Phi}_{AA'}\bra{\Phi})\right],\\
    &=(\mathbf{id}\otimes\mathcal{F})~J^{\mathcal{S}}_{AA'}
\end{align*}
where $J^{\mathcal{S}}_{AA'}$ is the CJ operator of channel $\mathcal{S}$. Now, since $\mathcal{S}$ is an NPT-breaking channel, we have $J^{\mathcal{S}}_{AA'} \in \mathbb{PPT}$ . Since local operations are a subset of PPT-preserving operations, the action of a local CP map on a PPT state gives a positive operator, which is PPT. Therefore, $\mathcal{F}\circ\mathcal{S}$ is an NPT-breaking CP map.
\end{proof}

Lemmas \ref{lemma3} and \ref{lemma4}, along with the convexity and compactness of the set $\mathbb{NPTBC}$ enable the certification of all non-NPT-breaking channels with minimal assumptions via semiquantum signaling games.

The certification follows the same general framework as the certification of non-$k$-SNB channels discussed in Sec.~\ref{result}, differing only in the choice of the payoff function. This difference arises because the witness operators for non-NPT-breaking channels are different than those used for non-$k$-SNB channels. The payoff function must satisfy 
\begin{align}
\left.\begin{aligned}
    &\mathscr{\tilde{J}}_{\text{avg}} \geq 0,~~\forall~\mathcal{S}\in \mathbb{NPTBC},\\    
\text{and}\quad &\mathscr{\tilde{J}}_{\text{avg}} <0,~~\text{when}~\mathcal{R}\not\in \mathbb{NPTBC}
\end{aligned}\right\}.
\end{align}

We have the following theorem, which is similar to Theorem \ref{Theo1}. 
\begin{theorem} \label{theo2}
    For every channel $\mathcal{R} \notin \mathbb{NPTBC}$, there exists a semiquantum signaling game for which the average payoff satisfies $\mathscr{\tilde{J}}_{\text{avg}}(\mathcal{R}) < 0$, while $\mathscr{\tilde{J}}_{\text{avg}}(\mathcal{S}) \geq 0$ for all $\mathcal{S} \in \mathbb{NPTBC}$.
\end{theorem}

\begin{proof}
    The proof is exactly similar to that of Theorem \ref{Theo1}, except that the payoff function $\mathscr{\tilde{J}}(b,x,y)$ satisfies
    \begin{align*}
    p(\psi_A^x)\, p(\phi_B^y)\, \mathscr{\tilde{J}}(b,x,y) =
    \begin{cases}
        \tilde{\gamma}_{x,y} & \text{if } b = 0, \\
        0 & \text{otherwise}
    \end{cases}
\end{align*}
where $\tilde{\gamma}_{x,y} \in \mathbb{R}$.
\end{proof}

An explicit illustration of this is provided in the example below.
\begin{example} 
    Consider a channel acting on $\mathcal{L}(\mathbb{C}^3)$ given by \cite{horodecki2000binding}
    \begin{equation} \label{channel}
        \Lambda(\cdot)=\frac{2}{7}(\cdot)+\frac{\alpha}{7} \sum_{i=1}^{3}P_{i\oplus1 i}(\cdot)P_{i  i\oplus1}+\frac{5-\alpha}{7} \sum_{i=1}^{3}P_{i\ominus1 i}(\cdot)P_{i  i\ominus1}
    \end{equation}
    where $P_{ij}=\ket{i}\bra{j}$, $\oplus$ and  $\ominus$ denote the $+$ and $-$ modulo $3$ respectively, and $2 \le \alpha \le 5$.
\end{example}
The Choi operator corresponding to this channel is the two-qutrit state given by \cite{PhysRevLett.82.1056}
\begin{equation}
    J^{\Lambda}=\frac{2}{7} \ket{\Phi_3}\bra{\Phi_3}+\frac{\alpha}{7}\sigma_+\frac{5-\alpha}{7}\sigma_-
\end{equation}
where 
\begin{equation*}
    \sigma_+=\frac{1}{3}(\ket{01}\bra{01}+\ket{12}\bra{12}+\ket{20}\bra{20}),
\end{equation*}
\begin{equation*}
    \sigma_-=\frac{1}{3}(\ket{10}\bra{10}+\ket{21}\bra{21}+\ket{02}\bra{02}),
\end{equation*}
\begin{equation*}
    \ket{\Phi_3}=\frac{1}{\sqrt{3}}(\ket{00}+\ket{11}+\ket{22}).
\end{equation*}

We are interested in finding the range of $\alpha$ for which the channel given by Eq.~\eqref{channel} is non-NPT-breaking. 

We use the following decomposable witness
\begin{equation}
    \tilde{W}=2 \ket{\Phi_3}\bra{\Phi_3}-3 \sigma_+
\end{equation}
This witness can be expressed as
\begin{equation}
    \tilde{W} = \sum_{x,y=1}^9\tilde{\gamma}_{x, y}~(\xi^{x}\otimes\xi^{y}),
\end{equation}

where the coefficients \( \tilde{\gamma}_{x,y} \) are represented in the form of a matrix given by
\begin{align} \label{2ndpayoff}
    \{\tilde{\gamma}_{x,y}\}_{x,y=1}^9 \equiv\begin{pmatrix}
        \frac{2}{3} & -1 & 0 & -\frac{1}{3} & -\frac{1}{3} & 0 & \frac{1}{3} & \frac{1}{3} & 0 \vspace{0.2cm}\\ 
        0 & \frac{2}{3} & -1 & -\frac{1}{3} & 0 & -\frac{1}{3} & \frac{1}{3} & 0 & \frac{1}{3} \vspace{0.2cm}\\
        -1 & 0 & \frac{2}{3} & 0 & -\frac{1}{3} & -\frac{1}{3} & 0 & \frac{1}{3} & \frac{1}{3} \vspace{0.2cm}\\
        -\frac{1}{3} & -\frac{1}{3} & 0 & \frac{1}{3} & 0 & 0 & 0 & 0 & 0 \vspace{0.2cm}\\
        -\frac{1}{3} & 0 & -\frac{1}{3} & 0 & \frac{1}{3} & 0 & 0 & 0 & 0 \vspace{0.2cm}\\
        0 & -\frac{1}{3} & -\frac{1}{3} & 0 & 0 & \frac{1}{3} & 0 & 0 & 0 \vspace{0.2cm}\\
        \frac{1}{3} & \frac{1}{3} & 0 & 0 & 0 & 0 & -\frac{1}{3} & 0 & 0 \vspace{0.2cm}\\
        \frac{1}{3} & 0 & \frac{1}{3} & 0 & 0 & 0 & 0 & -\frac{1}{3} & 0 \vspace{0.2cm}\\
        0 & \frac{1}{3} & \frac{1}{3} & 0 & 0 & 0 & 0 & 0 & -\frac{1}{3}
    \end{pmatrix}
\end{align}
where the values of $x(y)$ are denoted by the rows (columns) of this matrix, and $\xi^{x(y)}$ is given by Eq.~\eqref{state}. Since the inputs to Alice and Bob are the same as for the certification of non-$k$-SNB channels, the certification of non-NPT-breaking channels follows a similar circuit implementation, as discussed in Sec.~\ref{implement}. The coefficients $\tilde{\gamma}_{x,y}$ given by Eq.~\eqref{2ndpayoff} along with the payoff function 
{\[\mathscr{\tilde{J}}(b,x,y) = \left\{
        \begin{array}{ll}
            81\tilde{\gamma}_{x,y} & \text{if } b=0, \\
            0 & \text{otherwise}
        \end{array},
    \right.\] }gives the average payoff $\mathscr{\tilde{J}}_{\text{avg}}(\Lambda) < 0$ for $4 < \alpha \le 5$. 
    This certifies that the channel given by Eq.~\eqref{channel} is non-NPT-breaking for $4 < \alpha \le 5$, as shown in Figure \ref{plot3}.
    \begin{figure}[t!]
\includegraphics[width=0.5\textwidth]{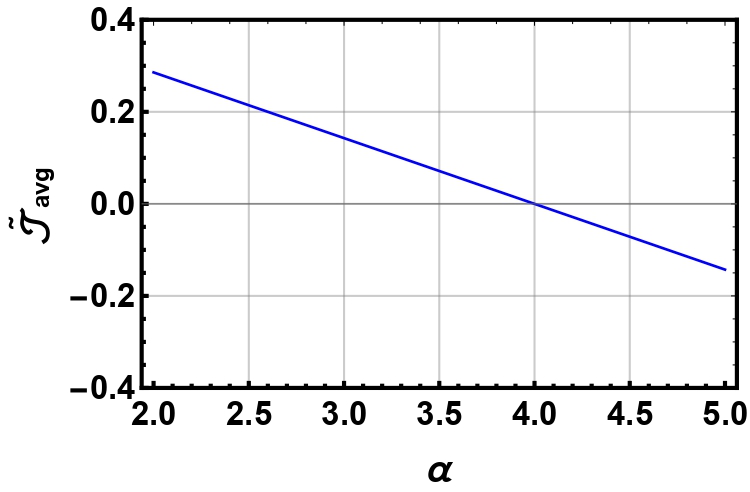} 
\caption{Variation of the average payoff $\mathcal{\tilde{J}}_{\text{avg}}$ with respect to the parameter ($\alpha$) for the channel given by Eq.~\eqref{channel}. A negative payoff denotes the range of parameters for which the channel is non-NPT-breaking.}\label{plot3}
\centering
\end{figure}

    Since there exist a decomposable witness operator corresponding to every non-NPT-breaking channel, our method can be applied for the certification of all non-NPT-breaking channels by choosing a suitable payoff function corresponding to that witness.
\\
\end{document}